\documentclass[letterpaper, 10 pt, conference]{ieeeconf}

\IEEEoverridecommandlockouts                              
\usepackage{cite}
\usepackage{amsmath,amssymb,amsfonts}
\usepackage{algorithm}
\usepackage{algpseudocode} 
\algtext*{EndWhile}
\algtext*{EndIf}
\algtext*{EndFor}

\usepackage{graphicx}
\usepackage{textcomp}
\usepackage{xcolor}

\usepackage{paralist} 
\usepackage{tikz}
\usetikzlibrary{shapes,arrows,fit,backgrounds,automata}
\usetikzlibrary{shapes.geometric,fit}
\usepackage[utf8]{inputenc}
\tikzset{main node/.style={circle,draw,minimum size=1cm,inner sep=0pt},}
\tikzstyle{container} = [draw, rectangle, inner sep=0.1cm]
\tikzstyle{container-ellipse} = [draw, ellipse, inner sep=0.1cm]
\usepackage{acronym}
\usepackage{amsthm}

\definecolor{darkgreen}{rgb}{0,0.5,0}

\algdef{SE}[DOWHILE]{Do}{DoWhile}{\algorithmicdo}[1]{\algorithmicwhile\ #1}%

\newcommand{\ie}{\textit{i.e.}}

\newcommand{\senseAct}{\Sigma}
\newcommand{\ctrlAct}{A}
\newcommand{\attackAct}{\act_2}
\newcommand{\calO}{\mathcal{O}}
\newcommand{\obs}{\mathsf{Obs}}
\newcommand{\dist}{\mathcal{D}}
\newcommand{\act}{\mathsf{A}}
\newcommand{\sense}{\sigma}
\newcommand{\att}{\beta}
\newcommand{\supp}{\mathsf{Supp}}
\newcommand{\plays}{\mathsf{Plays}}

\newcommand{\pre}{\mathsf{Pre}}
\newcommand{\post}{\mathsf{Post}}
\newcommand{\prefplays}{\mathsf{Prefs}}

\newcommand{\win}{\mathsf{Win}}

\newtheorem{theorem}{Theorem} 
\newtheorem{definition}{Definition}

\newtheorem{problem}{Problem}

\newtheorem{example}{Example}

\acrodef{pomdp}[POMDP]{partially observable Markov decision process}
\acrodef{mdp}[MDP]{Markov decision process}
\acrodef{asw}[ASW]{Almost-Sure Winning}
\acrodef{cps}[CPS]{Cyber-Physical Systems}
\title{Qualitative Planning in Imperfect Information Games with Active Sensing and Reactive Sensor Attacks: Cost of Unawareness}

\author{Abhishek N. Kulkarni, Shuo Han, Nandi O. Leslie, Charles A. Kamhoua, and Jie Fu$^\ast$
\thanks{A. Kulkarni and J. Fu ($^\ast$ corresponding author) are with the Robotics Engineering Program and Dept. of Electrical and Computer Engineering, Worcester Polytechnic Institute, Worcester, MA 01609 USA.
{\tt\small \{ankulkarni,jfu2\}@wpi.edu}}
\thanks{S. Han is with the Department of Electrical and Computer Engineering, University of Illinois, Chicago, IL 60607.
{\tt\small hanshuo@uic.edu}}
\thanks{N. Leslie is with Raytheon Technologies. 
{\tt\small nandi.o.leslie@raytheon.com}}
\thanks{C. Kamhoua is with U.S. Army
Research Laboratory. 
{\tt\small charles.a.kamhoua.civ@mail.mil}}
}

\begin{document}
\maketitle
\thispagestyle{empty}
\pagestyle{empty}

\begin{abstract}
We consider the probabilistic planning problem  where the agent (called Player 1, or P1) can jointly plan the control actions and sensor queries in a sensor network and an attacker (called player 2, or P2) can carry out attacks on the sensors. We model such an adversarial interaction using a formal model--a reachability   game with partially controllable observation functions. The main contribution of this paper is to assess the cost of P1's unawareness: Suppose P1 misinterprets the sensor failures as probabilistic node failures due to unreliable network communication, and P2 is aware of P1's misinterpretation in addition to her partial observability. Then, from which states can P2 carry out sensor attacks to ensure, with probability one, that P1 will not be able to complete her reachability task even though, due to misinterpretation, P1 believes that she can almost-surely achieve her task. We develop an algorithm to solve the almost-sure winning sensor-attack strategy given P1's observation-based strategy. Our attack analysis could be used for attack detection in wireless communication networks and the design of provably secured attack-aware sensor allocation in decision-theoretic models for cyber-physical systems.
\end{abstract}


\section{Introduction}
    \label{sec:introduction}

Security of \ac{cps} has been studied extensively in the systems and control community (see a survey in \cite{dibajiSystemsControlPerspective2019}). In these control-theoretic approaches, the system under consideration is modeled as a linear or nonlinear dynamical system. The malicious attacker on sensors and actuators can inject errors into the sensor measurement \  (such as spoofing attacks or false data injection attacks), block or delay signals (as in denial of services and jamming attacks) \cite{Xing2010,cetinkayaOverviewDenialofServiceAttacks2019}. Detection and secured estimation and control under malicious attacks have been investigated to ensure  resilience, stability, and robustness of the control system. Secured state estimation under sensor and/or actuator attacks has been studied for linear time-invariant systems \cite{shoukrySecureStateEstimation2017,changSecureEstimationBased2018,showkatbakhshSecuringStateReconstruction2020}.

While a plethora of work analyze \ac{cps} security from a systems and control perspectives, discrete event systems and stochastic games have also been employed to analyze secured control under attacks. The authors in    \cite{niuOptimalSecureControl2020} model the attacker-defender interaction  as a turn-based stochastic Stackelberg game, where the attacker observes the defender's strategy before deciding how to carry out attacks  and the defender is to maximize the probability of satisfying a temporal logic formula. They assume that both players have complete observations and the game state is controlled by the joint attacker and defender actions. In supervisory control, the supervisor aims to control the behaviors of a DES to be within a desired language while an attacker can add, delete, or replace symbols (sensor inputs or actuator outputs). The authors  in \cite{wangSupervisoryControlDiscrete2019} propose  finite-state transducers as attack models, which are known to the supervisor.  A composition of the plant and the  attack transducers  is generated to evaluate, whether the supervisor can still control the plant's outputs within the desired language despite of the malicious attacks.  In \cite{wakaikiSupervisoryControlDiscreteEvent2019}, the authors do not assume that attacker's policy is known but instead introduce a class of multi-adversary games where the supervisor is played  against multiple possible attackers on sensors.  They proposed the condition for controllability and observability under  such attacks and the supervisory control design to achieve the control objective robustly.
Besides mitigating sensor/actuator attacks, opacity in DESs \cite{dubreilSupervisoryControlOpacity2010,sabooriOpacityEnforcingSupervisoryStrategies2012,wuSynthesisInsertionFunctions2014} is  to mitigate attacks on confidentiality of the systems.

Our work is motivated by the development of networked robotic systems and the security issues. Networked robots are a class of \ac{cps}s that focus on the seamless integration of robots  to assist humans in difficult tasks such as security paroling and contested search and rescue. For these applications, the robots operate in an uncertain and potentially adversarial environment, subject to malicious attacks in cyber and physical spaces. In this scope, we focus on sensor attack in networked robots and employ decision-theoretic models to nvestigate how the sensor attack can affects an intelligent agent's information states, beliefs, and consequentially its decisions. Consider a robot is to achieve a task in
a stochastic environment with partial observations. At run-time, the robot can \emph{actively} query a subset
of sensors in a wireless  network to update his belief about the state and the progress with
respect to the task. However, the attacker can choose unsafe sensor
nodes and carry out reactive jamming attacks \cite{cetinkayaOverviewDenialofServiceAttacks2019}. 
 Such an attack will target the active communication and 
 block the sensor information  from reaching the agent. 
Our key questions are, if the robot
mistakes  sensor failures as probabilistic node failures in the wireless network \cite{wangReliabilityWirelessSensor2017}, how would the
robot compute its active sensing and control policy to reach a goal state with probability one, \ie almost-surely win? How shall the attacker
exploit the robot's misinterpretation of the sensor failures for carrying out successful sensor attacks that manipulate the robot's information states?  With the knowledge about attacker's almost-sure winning region given the naive robot, we can understand how the physical task completion is coupled with the vulnerabilities in the cyber network and gain critical insight about hardening network security with task-oriented sensor allocation and attack-aware active sensing strategies. Our technical approach employs a formal model called stochastic game on graphs \cite{chatterjeePartialobservationStochasticGames2012,chatterjeeSurveyStochasticOregular2012,bloemGraphGamesReactive2018}. The key contributions are the following: \begin{inparaenum}[1)]
	\item We formulate  a class of \ac{pomdp}s with reachability objectives. The model generalize from regular POMDPs to include both actions and sensor queries as the decisions for the robot.
	\item We extend probabilistic model checking in POMDPs to answer, from which set of information states, the robot has an observation-based control and active sensing strategy that ensures the task can be achieved with  probability one given probabilistic action outcomes and probabilistic sensor failures. 
	\item  We analyze the adversary's game to compute, given the naive robot's strategy, when to attack which sensors so that the robot cannot achieve the task with probability one, even from the robot believes it can do so.  
\end{inparaenum} The solution approach is illustrated using a running example.

Our modeling and design approaches differ from existing literature on secured control design in three key aspects: 1) In our model, the system dynamics is probabilistic instead of nondeterministic (as studied in DESs) or linear, time-invariant (as studied in secured control in \ac{cps}s).  2) Our    objective is to synthesize the almost-sure winning strategy for the robot and use it to assess the almost-sure winning region for the attacker who exploits the robot's unawareness of attacks. This criteria of performance is different from controllability and observability in DESs. Due to the modeling and objective differences, our solution approaches differ greatly from that being used in DESs and linear systems. 3)  We do not assume the knowledge of attack strategies. Instead, we analyze the game and compute the attack strategy given the attacker's information and assumption of the agent's perception. This formal synthesis provides security insights considering the worst-case attacker strategy. 4) In relation to other game-theoretic based secured control in \ac{cps}s, we   focus on sensor attacks given partial observations and active sensor queries, instead of actuator attacks with perfect observations  \cite{niuOptimalSecureControl2020}.

\section{Problem Formulation}
    \label{sec:problem}
    
We consider the probabilistic planning of an autonomous agent (called Player 1, or P1, pronoun `she') in an adversarial environment, where an attacker (called player 2, or P2, pronoun `he') can carry out attacks on the sensors of P1. The objective of P1 is to reach a set of goal states. Such an interaction between P1 and P2 is captured using the following model, where  P1's observation function is partially controllable.

\begin{definition}[Zero-sum Stochastic Reachability Game with Partially Controllable Observation Function]
    \label{def:system}
    A two-player stochastic game with partially controllable observation function in which P1 has a reachability objective is a tuple 
    \[
        G = \langle S, \ctrlAct, P, s_0, \Gamma, \senseAct \times \attackAct, \calO, \obs, o_0, F \rangle 
    \]
    where 
    \begin{itemize}
        \item $\langle S, \ctrlAct, P, s_0 \rangle$ is a \ac{mdp} where $S$ is a finite set of states; $A$ is a finite set of actions; $s_0$ is an initial state; $P: S \times A \times S \rightarrow [0, 1]$ is a transition probability function such that $\sum_{s' \in S}P(s, a, s') = 1$ for any state $s \in S$ and action $a \in A$;
        
        \item  $\Gamma= \{0,1,\ldots, N\}$ is a set of indexed sensors. Sensor $i$ covers a subset of states. The function $\gamma: \Gamma \rightarrow 2^S$ maps a sensor to a set of states covered by that sensor.
        
        \item $\senseAct\subseteq 2^\Gamma$ is a set of sensor query actions of P1, each of which requests a sensor reading from a unique subset of sensors from $\Gamma$;

        \item $\attackAct \subseteq 2^\Gamma$ is a set of sensor attack actions of P2, each of which jams a sensor reading of a subset of sensors from $\Gamma$;
        
        \item $\calO \subseteq 2^S$ is a finite set of observations;
        
        \item $\obs: S \times \senseAct \times \attackAct \rightarrow \calO$ is a deterministic observation function, which maps a state $s \in S$, a sensor query action $\sense$ and a sensor attack action $\att$ into an observation $o = \obs(s, \sense, \att) \in \calO$. Two states $s, s' \in S$ are said to be observation equivalent given the sensing action and sensor attack action $\sense,\att$ if $\obs(s,\sense,\att)= \obs(s',\sense,\att)$. The observation $o$ includes the set of observation-equivalent states.
        
        \item $o_0 =\{s_0\} \in \calO$ is an initial observation;
        
        \item $F \subseteq S$ is a set of final states that P1 must reach to complete her task. 
    \end{itemize}
\end{definition}

To simplify the exposition of concepts in this paper, we consider the sensors in $\Gamma$ to be Boolean sensors. That is, each sensor $i$ returns a Boolean value, denoted $v_i$: $v_i=\top$ if the states covered by the sensor contains the \emph{current} state of $G$ and $v_i=\bot$ otherwise. Whenever a sensor is attacked, it returns a special value $v_i=?$ denoting a sensor failure. Given a set $\Gamma$ of sensors, the set of observations $\calO$ can be constructed by using \textsc{GetObservation} in Alg.~\ref{alg:deduction} for each state $s \in S$, sensing action $\sense \in \senseAct$ and $\att \in \attackAct$.

\begin{algorithm}[t]
\caption{\textsc{GetObservation}}
\label{alg:deduction}

\begin{algorithmic}[1]
\item[\textbf{Inputs:}] $s, \sense, \att$
\item[\textbf{Outputs:}] $V$  
    \State $V = S$ 
    \ForAll{$i \in \sense \setminus \att$}
        \If{$v_i=\top$}
            \State $V \gets V \cap \gamma(i)$
        \Else
            \State $V \gets V \cap (S \setminus \gamma(i))$
        \EndIf
    \EndFor
    \State \Return $V$
\end{algorithmic}
\end{algorithm}

It is noted that in the absence of sensing attacks, the game is a \ac{pomdp} with active sensing actions. With sensor attacks, the observation of P1 is partially controllable. 

\textbf{Game Play.} The game play in $G$ is constructed as follows. From the initial state $s_0$, P1 obtains the initial observation $o_0$. Based on the observation, P1 selects a control action $a_0 \in \ctrlAct$ and a sensor query action $\sense_0 \in \senseAct$. The system moves to state $s_1$ with probability $P(s_0, a_0, s_1)$. At state $s_1$, P2 selects an attack action $\att_0 \in \attackAct$. The system generates a new observation $o_1 = \obs(s_1,\sense_0,\att_0)$. This process repeats until P1 \emph{knows}, with probability one, that a state in $F$ is reached. We denote the resulting play as $\rho = s_0 a_0 \sense_0 \att_0 s_1 a_1 \sense_1 \att_1 \ldots$. The set of infinite plays in $G$ is denoted by $\plays(G)$ and the set of finite prefixes of plays is denoted by $\prefplays(G)$. We say a state $s \in S$ to be \emph{reachable} in $G$ if there exists a prefix $\rho \in \prefplays(G)$ such that $s \in \rho$. We denote the set of possible reachable states given a state-action pair $(s,a)$ by  $\post_G(s, a)$. Formally,   $\post_G(s, a)= \{s'\in S \mid P(s, a, s') > 0\} $. Abusing the notation, we define $\post_G(B, a) = \bigcup_{s \in B} \post_G(s, a)$ where $B\subseteq S$ is a subset of states.

\textbf{Strategies.} Given the fact that P1 must determine, simultaneously, an action $a\in A$ and a sensor query action $\sense \in \senseAct$, we denote $\act_1  = A \times \senseAct$ as the action space of P1 and $\act_2= \attackAct$  as the action space of P2. A history-dependent, randomized (resp., deterministic) strategy for player $j \in \{1, 2\}$ is a function $\pi_j : \prefplays(G) \rightarrow \dist(\act_j)$ (resp., $\pi_j : \prefplays(G) \rightarrow \act_j$). We say that player $j$ follows strategy $\pi_j$ if for any prefix $\rho\in \prefplays(G)$ at which $\pi_j$ is defined, player $j$ takes the action $\pi_j(\rho)$ if $\pi_j$ is deterministic, or an action $a \in \supp(\pi_j(\rho))$ with probability  $\pi_j(\rho,a)$ if $\pi_j$ is randomized. The set of all strategies of player $j$ is denoted by $\Pi_j$. 

A strategy $\pi_1 \in \Pi_1$ is said to be \emph{almost-sure winning} for P1 over a reachability objective to visit $F$ if, for any $\pi_2 \in \Pi_2$, P1 is guaranteed to visit $F$ with probability 1. Similarly, a strategy $\pi_2$ is almost-sure winning for P2 against P1 if, for any $\pi_1\in \Pi_1$, P2 ensures with probability one that no state in $F$ can be visited. A state is called an \emph{almost-sure winning} state for a player, if there exists an almost-sure winning strategy for the player at that state.

  \textbf{Observation Equivalence.} Given the observation function $\obs$, an \emph{observation} of a play $\rho$ is defined as $\obs(\rho)= o_0 (a_0, \sense_0,\att_0) o_1 (a_1, \sense_1, \att_1)\ldots $ where $o_{i+1}=\obs(s_{i+1}, \sense_i, \att_i)$ for all $i\ge 0$ and $o_0$ is the initial observation. Two plays (or play prefixes) $\rho, \rho'$ are said to be observation-equivalent, denoted by $\rho \sim \rho'$. A strategy is said to be \emph{observation-based} if $\pi_j(\rho) = \pi_j(\rho')$ whenever $\rho \sim \rho'$. We denote the set of all observation-based strategies of P1 by $\Pi_1^O$.

\textbf{Information Structure.} In this paper, we consider a game with one-sided partial observation, in which P1 has partial and P2 has perfect observations. Thus, the adversarial interaction in game $G$ is characterized as follows. During her turn, P1 uses the sensing action to reduce the uncertainty in her belief about the \emph{current} state. Whereas P2, who knows the current state, uses attack actions to control how much information P1 gains about the \emph{current} state of the game. To make the problem nontrivial, we also consider that the attack actions can be limited.

We assume the information about the game is asymmetric and incomplete for P1. Specifically, P1 considers the failures of a subset of sensors $\att \in \attackAct$ to be probabilistic failures. Only P2 has the correct, complete information of the game. Hereafter, we refer to such P1 as a \emph{na\"ive P1}. Our main goal is to understand the cost of such unawareness of sensor attacks: \emph{whether it is possible that P2 can win from a state that P1 believes to be almost-surely winning for her?} We formalize our problem as follows.

\begin{problem}
    Given the stochastic game $G$ with partially controllable observation function in which P1 has partial and P2 has perfect observability, determine when (for which prefix in $\prefplays$) there exists an \emph{observation-based} strategy $\pi_1 \in \Pi_1^O$ using which the na\"ive P1 believes that she can satisfy a reachability objective over $F$ with probability one, no matter which strategy P2 plays. Also, determine when there exists a P2's best attack strategy that prevents the na\"ive P1 from satisfying her reachability objective.
\end{problem}

\section{Main Result}
    \label{sec:type3naive}

In this section, we present a \emph{qualitative} analysis of game $G$, in which we show how to compute the almost-sure winning strategies of na\"ive P1 who has partial observability.
We start by introducing a running example to illustrate the advantage of sensing actions and the effect of sensor attacks.

\begin{example}[Part I] \label{ex:pt-1}
    Fig.~\ref{fig:ex-G} represents a two-player stochastic game with partially controllable observation function in which $s_0$ is the initial state and $s_4$ is the final state. The set $\Gamma$ consists of 4 sensors $0$ (red), $1$ (blue), $2$ (green), $3$ (violet) covering states $\{s_0, s_1\}, \{s_1, s_2\}, \{s_0, s_2, s_3\}$ and $\{s_4, s_5\}$, respectively. P1 has three control actions, $\ctrlAct = \{a_0, a_1, a_2\}$ and three sensing actions $\sense_0, \sense_1$ and $\sense_2$ which query the sensors $\{0,1\}, \{0,2\}$ and $\{2\}$, respectively. P2 has three attack actions, $\att_0, \att_1$ and $\att_2$ that jam the sensors $0,1$ and $2$, respectively. The initial observation $o_0=\{s_0\}$. Given our focus is on qualitative analysis, the probabilistic transitions are labeled with actions only. For instance, at state $s_0$, the action $a_0$ can be taken and reach any of the states $s_1$, $s_2$, and $s_0$ with a positive probability. The edges from $s_0$ labeled $a_0$ are the probabilistic outcomes given that action.
    
    \begin{figure}
        \centering
       \begin{tikzpicture}[->,>=stealth',shorten >=1pt,auto,node distance=2.5cm, scale =0.8,transform shape]

  \node[state, initial] (0) {$s_0$};
  \node[state] (2) [right of=0] {$s_2$};
  \node[state] (1) [above of=0] {$s_1$};
  \node[state] (3) [right of=2] {$s_3$};
  \node[state, accepting] (4) [above of=2] {$s_4$};
  \node[state] (5) [above of=3] {$s_5$};
   \node [container,fit=(0) (1),draw=red,dashed,line width=0.2mm  ] (container) {};
   \draw [rotate=45,draw=blue, dashed,line width=0.2mm] (1.75,0) ellipse (0.75cm and 2.5cm);
    \node [container,fit=(0) (2) (3) ,draw=green, dashed,line width=0.2mm] (container) {};
    
    \node [container,fit=(4) (5) ,draw=violet, dashed,line width=0.2mm] (container) {};
  \path (0) edge[loop below]  node {$a_0, a_1$} (0)
        (0) edge[bend left]   node {$a_0$} (1)
        (0) edge[bend left]   node {$a_0$} (2)
        
        (0) edge[bend right]  node[pos=0.2] {$a_1$} (1)
        (0) edge[bend right]  node {$a_1$} (2)
        
        (0) edge              node[pos=0.2] {$a_2$} (2)
        (0) edge[bend right, out=-60, in=240]  node {$a_3$} (3)
        
        (1) edge              node {$a_0$} (4)
        (1) edge[bend left, out=60, in=120]   node {$a_1$} (5)
        
        (2) edge              node {$a_0$} (5)
        (2) edge              node[pos=0.2]{$a_1$} (4)
        
        (3) edge[bend left]   node {$a_0$} (5)
        (3) edge[bend right]  node {$a_1$} (5)
  ;

\end{tikzpicture}
        \caption{A two-player stochastic game with partially controllable observation function. The dashed region represent the sensors: $0$ (red), $1$ (blue), $2$ (green), $3$ (violet). }
        \label{fig:ex-G}
    \end{figure}
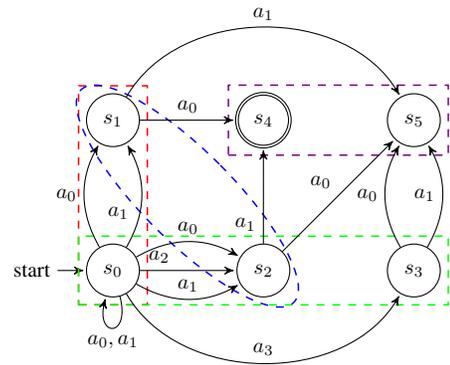
    
    To illustrate the advantage of sensing actions, we first analyze the game in Fig.~\ref{fig:ex-G} when P1 has no sensing actions and P2 has no attack actions, P1 has no almost-sure winning strategy at $s_0$. This is because choosing the action $a_2$ at $s_0$ is unsafe as it may lead the game to the losing state $s_5$ with positive probability. And, by choosing $a_0$ or $a_1$, P1's new belief state would be $B = \{s_0, s_1, s_2\}$, at which she does not have a consistent winning action, \ie~if the true game state is $s_1$ then P1 should choose action $a_0$ and consider action $a_1$ to be unsafe. However, if the true game state is $s_2$ then action $a_0$ is unsafe. Not knowing which state is the true state, P1 cannot select any action without risking running into the sink state $s_5$.
    
    On the contrary, when P1 can query sensors but P2 cannot attack any of them (the game is still a \ac{pomdp}), she has an almost-sure winning strategy at $s_0$ by selecting, say, $a_0$ as control action and $\sense_0$ as sensing action. This is because the sensing action $\sense_0$ allows P1 to determine the resulting state after choosing the action $a_0$.
  
\end{example}

Next, we analyze how P1 can synthesize her almost-sure winning active sensing and control strategies when she thinks the sensors may have probabilistic failures.

\subsection{The Game Perceived by P1}  

We are interested to synthesize P1's active sensing and control strategies to reach $F$ with probability one, when she mistakes sensor failures as probabilistic node failures. Due to P1's incorrect interpretation of sensor failures, from P1's perspective, $G$ is a \ac{pomdp} with active sensing actions.

\begin{definition}

The game $G$ as perceived by na\"ive P1 is the tuple,
    \[
        G^1 = \langle S, \act_1 = \ctrlAct \times \senseAct, P, \calO, \obs^1, o_0, F\rangle,
    \]
    where $P$ is the same probabilistic transition function as in game $G$.     The probabilistic observation function $\obs^1(s, \sense, o)$ is defined such that, given a state $s \in S$ and P1's sensing action $\sense \in \senseAct$, the probability of obtaining an observation $o \in \calO$ is strictly positive, if there exists an attack action $\att \in \attackAct$ enabled at the state $s$, such that $\obs(s, \sense,\att) = o$.
    \end{definition}

To derive an almost-sure winning strategy for P1 in $G^1$, we construct an \ac{mdp}  with perfect observation, in which the belief of P1 about the current state is made explicit. Our construction is adopted from reachability analysis in POMDP \cite{baierDecisionProblemsProbabilistic2008}. Also, for qualitaitive planning, we only need to know the support of a distribution for the next state given a state-action pair, but need not to know the exact distribution. 
\begin{definition}
    Given $G^1$, the perfect-observation \ac{mdp} of P1 is a tuple, 
    \[
        H = \langle Q \cup \{q_F\}, \act_1, \delta, q_0, q_F \rangle,
    \]
    where 
    \begin{itemize}
        \item $Q = \{(s, B) \mid \exists o \in \calO \text{ s.t. } B \subseteq o \text{ and } s \in B \}$ is the set of states;
        
        \item $q_F$ is a single final state. It is also a sink state.
        
        \item $q_0 = (s_0, \{s_0\})$ is the initial state;
        
        \item Given a state $(s, B)$ and action $(a, \sense)$, the transition function $\delta$ is defined by,
        \begin{itemize}
            \item If $\post_G(s,a)\cap F = \emptyset$ then  $\delta((s, B), (a, \sense), (s', B')) > 0$  if and only if $s' \in \post_G(s, a)$ and there exists an $o \in \calO$ with $\obs^1(s', \sense, o) > 0$ such that $B' = \post_{G}(B, a) \cap o$;
            \item If $\post_{G}(s, a) \subseteq F$ then $\delta((s, B), (a, \sense), q_F) = 1$;
            
            \item If  $\post_G(s,a)\cap F\ne  \emptyset$ and $\post_G(s,a)\setminus F \ne  \emptyset$, then $\delta((s, B), (a, \sense), q_F) > 0$ and $\delta((s, B), (a, \sense), (s', B')) > 0$ for each $s' \in \post_{G}(s, a)$ such that there exists an $o \in \calO$ with $\obs^1(s', \sense, o) > 0$ such that $B' = \post_{G}(B, a) \cap o$.
        \end{itemize}

    \end{itemize}
\end{definition}

The transition function can be understood as follows. Given a state $(s, B)$ and action $(a, \sense)$, if  $\post_{G}(s, a) \subseteq F$, then the game reaches the final state $q_F$ with probability one. If none of the states $s' \in \post_G(s, a)$ is a final state, then with a sensing action $\sense$ and an observation $o$, the game reaches each state $(s', B')$ with a positive probability for which $P(s, a, s') > 0$ and the belief $B'$ is \emph{consistent} with the observation $o$. Lastly, if there exists $s'\in \post_G(s, a) \setminus F$ and $\post_G(s,a)\cap F\ne \emptyset$, then the game reaches $q_F$ and all  $(s', B')$ where $s'$  is reachable from $s$ with a positive probability, and the belief $B'$ is consistent with some observation.

As the belief is constructed only using observations, P1, after observing two observation-equivalent play $\rho$ and $\rho'$ will generate the same belief $B$. Thus, we say two states $q = (s, B)$ and $q' = (s', B')$ are (observation-)equivalent to P1 whenever $B = B'$. We denote the equivalence of two states by $q \sim q'$ and the set of all states equivalent to $q$ by $[q]_\sim$. A memoryless\footnote{Memoryless strategies are sufficient for qualitative analysis of \ac{pomdp}s \cite{baierDecisionProblemsProbabilistic2008}.} randomized strategy of P1 $\pi: Q  \rightarrow \dist(\ctrlAct \times \senseAct)$ is said to be \emph{equivalence-preserving} if and only if $\pi(q) = \pi(q')$ whenever $q \sim q'$.     

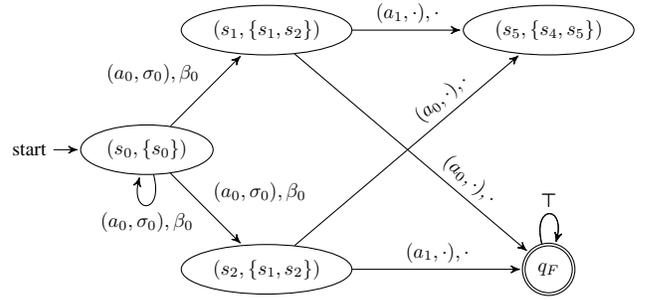
\begin{figure}[t]
        \centering
        \begin{tikzpicture}[->,>=stealth',shorten >=1pt,auto,node distance=5cm,
        scale = 0.75,transform shape]
        \tikzstyle{every state}=[fill=white,draw,ellipse]

  \node[state, initial] (0) {$(s_0, \{s_0\})$};
  \node[state] (1) [above right of=0, node distance =3cm] {$(s_1, \{s_1, s_2\})$};
  \node[state] (2) [below right of=0,node distance =3cm] {$(s_2, \{s_1, s_2\})$};
  \node[state] (3) [right of=1] {$(s_5, \{s_4, s_5\})$};
  \node[state,accepting] (4) [right of=2] {$q_F$};

  \path 
    (0) edge[loop below]  node {$(a_0, \sense_0), \att_0$} (0)
    (0) edge  node {$(a_0, \sense_0), \att_0$} (1)
    (0) edge  node {$(a_0, \sense_0), \att_0$} (2)
    
    (1) edge  node[pos=0.7,sloped] {$(a_0, \cdot), \cdot$} (4)
    (1) edge  node {$(a_1, \cdot), \cdot$} (3)
    
    (2) edge  node[pos=0.7,sloped] {$(a_0, \cdot), \cdot$} (3)
    (2) edge  node {$(a_1, \cdot), \cdot$} (4)
    
    (4) edge[loop above]  node {$\top$} (4)
    (4) edge[loop above]  node {$\top$} (4)
  ;

\end{tikzpicture}
        \caption{A subset of perfect-observation \ac{mdp} constructed by P1 in her mind. The figure shows the relevant states and edges when P1 chooses $(a_0, \sense_0)$ at the initial state $(s_0, \{s_0\})$ and sensor $0$ fails.}
        \label{fig:ex-H}
\end{figure}

\setcounter{example}{0}
\begin{example}[Part II]
    Fig.~\ref{fig:ex-H} shows a subset of perfect-observation \ac{mdp} that P1 constructs in her mind when she chooses the action $(a_0, \sense_0)$ at initial $q_0 = (s_0, \{s_0\})$ and the sensor $0$ (corresponding to attack action $\att_0$) fails. The game may reach any of the states $(s_0, \{s_0\})$, $(s_1, \{s_1, s_2\})$ or $(s_2, \{s_1, s_2\})$ with a positive probability. For instance, if the true state transitions from $s_0$ to $s_1$, then P1 would obtain an observation $o = \{s_1, s_2\}$ as sensor $0$ is attacked and the belief of P1 is updated to $B' = \post_G(s_0, a_0) \cap o = \{s_0, s_1, s_2\} \cap \{s_1, s_2\} = \{s_1, s_2\}$. When P1 chooses an action $(a_0, \cdot)$ at $(s_1, \{s_1, s_2\})$, where $\cdot$ can be any sensing action, the game is ensured to reach the final state regardless of any sensor failure (because $\post_G(s_1, a_0) = \{q_F\}$). Similarly, when P1 chooses an action $(a_1, \cdot)$ at $(s_1, \{s_1, s_2\})$, the game reaches the state $(s_5, \{s_4, s_5\})$. 
    
    However, it is noted that P1 cannot distinguish whether she is in $(s_1, \{s_1, s_2\})$ or $(s_2, \{s_1, s_2\})$. We observe that the action $(a_0, \sense_0)$, which was winning for P1 at the initial state $(s_0, \{s_0\})$ when no sensor failures are considered (see Example~\ref{ex:pt-1} Part I), is no longer winning for her. This is because the sensor failure results in P1's belief state to be $\{s_1, s_2\}$, at which she does not have an  action that almost-surely reaches $q_F$---as P1 cannot distinguish whether the state is $(s_1, \{s_1,s_2\})$ or $(s_2,\{s_1,s_2\})$.
    
\end{example}

Next, we present Alg.~\ref{alg:pomdp-reachability} to synthesize P1's almost-sure winning belief-based strategy. We denote by $\post_H(q, (a ,\sense)) =\{q' \mid \delta(q,(a,\sense))>0\}$. Given a state $q \in Q \cup \{q_F\}$, the set of its predecessors state-action pairs is denoted by 
\begin{multline*}
    \pre(q) = \{(p, (a, \sense))  \mid   q \in \post_{H}(p, (a,\sense))\}.
\end{multline*}
and generalize this operator to subsets of states, $\pre(X) = \cup_{q\in X} \pre(q)$. The algorithm starts by identifying the set $L_0$ of states from which there \emph{does not} exist a path (a sequence of transitions with positive probabilities) to reach the final state $q_F$. That is, the states in $L_0$ are clearly not almost-sure winning for P1. $L_0$ can be computed using standard graph algorithms over $H$. Subsequently, in $k$-th iteration, the algorithm identifies and eliminates those actions at predecessors $q$ of some $u \in L_k$ which visit $L_k$ in one step with a positive probability. As P1 cannot distinguish between equivalent states, if she removes an action $(a, \sense)$ from the set of allowable actions $\pi(q)$ at $q$ then she must also remove that same action for any state $p$ that is equivalent to $q$. As P1 does not have an almost-sure winning strategy from any state $q$ whose allowable actions set $\pi(q)$ is empty, such a state is added to $L_{k+1}$. The $k$ increments by one.  The process repeats until for some $k$, $L_{k+1} =\emptyset$.

\begin{algorithm}
\caption{Belief-based Strategy in POMDP Reachability }
\label{alg:pomdp-reachability}

\begin{algorithmic}[1]
\item[\textbf{Inputs:}] $H$ 
\item[\textbf{Outputs:}] $\win_1$: Almost-sure winning region of P1, $\pi$: belief-based almost-sure winning strategy of P1
    \State $L_0 \gets \{q \in  Q \mid q_F \text{ is not reachable from } q \}$
    \State For all $q \in L$, $\pi(q)=\emptyset$ and for all $q \in Q \setminus L$, $\pi(q)= \{(a,\sense)\mid \delta(q,(a,\sense)) \text{ is defined.}\}$
    \State $k\gets 0$
    \While{$L_k \neq \emptyset$}
        \State $L_{k+1} \gets \emptyset$
        \ForAll{$u \in L_{k}$}
            \ForAll{$(q, (a, \sense)) \in \pre(u)$}
                \ForAll{$p \in [q]_\sim$}
                    \State Remove $(a, \sense)$ from $ \pi(p)$
                    \If{$ \pi(p)  = \emptyset$ and $p \notin \bigcup_{i=0}^k L_i$}
                        \State Add $p$ to $L_{k+1}$
                    \EndIf
                \EndFor
            \EndFor
        \EndFor
        \State $k\gets k+1$.
    \EndWhile
    \State \Return $\win_1 = Q \setminus \bigcup_{i=0}^k L_i$, $\pi$.
\end{algorithmic}
\end{algorithm}

\begin{theorem}
    P1 has an almost-sure winning strategy in $G^1$ to visit $F$ if and only if $q_0 \in \win_1$, where $\win_1$ is the set of almost-sure winning states computed by Alg.~\ref{alg:pomdp-reachability}. 
\end{theorem}

\begin{proof}[Proof (Sketch)]
    It is known from \cite{baierDecisionProblemsProbabilistic2008} that P1 has an almost-sure winning strategy in $G^1$ to visit $F$ if and only if she has an almost-sure winning strategy in $H$ to visit $q_F$. Thus, for the statement to hold, it must be the case that Alg.~\ref{alg:pomdp-reachability} must identify every almost-sure winning state in $H$. This must be true because Alg.~\ref{alg:pomdp-reachability} removes $q$ from \ac{mdp} only if there is no safe action given P1's belief $B$ at $q=(s,B)$, i.e. for each enabled action, there exists a state $p\in [q]_\sim$ such that by taking that action, P1 may reach a state from $p$ from which $q_F$ is not reachable with a positive probability. 
\end{proof}

\textbf{ASW Strategy.} Given $\pi: \win_1\rightarrow 2^{\act_1}$ maps each state $q\in \win_1$ to a set of allowed actions (see Alg.~\ref{alg:pomdp-reachability}), P1's almost-sure winning strategy selects each action in $\pi(q)$ randomly with a positive probability. It is noted that P1's strategy $\pi$ is indeed a multi-strategy, or an infinite set of randomized strategies, because different choices of the probabilistic distributions given the set of states yield different randomized strategies. Again, for qualitative analysis, no matter which randomized strategy $\pi_1$ P1 selects, as long as the support   $\supp(\pi_1(q))= \pi(q)$ for any $q\in \win_1$, P1 can ensure almost-sure winning in her perceived POMDP with probabilistic sensor failures.

\setcounter{example}{0}
\begin{example}[Part III]
    Consider the subset of perfect-observation \ac{mdp} in Fig.~\ref{fig:ex-H}. As $s_4$ is unreachable from $s_5$, all states $(s_5, B)$ for any $B \subseteq S$ are contained in $L_0$. Thus, while investigating the state $(s_1, \{s_1, s_2\})$,  Alg.~\ref{alg:pomdp-reachability} identifies  action $(a_1, \sense_0)$ to lead to $L_0$. Consequently, the action $(a_1, \sense_0)$ is removed from states $(s_1, \{s_1, s_2\})$ and $(s_2, \{s_1, s_2\})$ as they are equivalent. Similarly, the action $(a_0, \sense_0)$ is removed from both states, thereby eliminating them from set of almost-sure winning states. In the next iteration, the initial state is also removed from set of almost-sure winning states---concluding that action $(a_0, \sense_0)$ is not almost-sure winning for P1.
    
    Next, consider Fig.~\ref{fig:ex-H2}, which shows a subset of perfect-observation \ac{mdp} when P1 chooses the action $(a_0, \sense_1)$ at initial $q_0 = (s_0, \{s_0\})$ and either the sensor $0$ or $2$ (corresponding to attack action $\att_0$ and $\att_2$) fails. These states are not eliminated by Alg.~\ref{alg:pomdp-reachability}. To understand the winning strategy, observe the equivalent states $(s_2, \{s_0, s_2\})$, at which the action $(a_0, \sense_1)$ reaches $q_F$ with positive probability or visits one of the states among $(s_1, \{s_1\}), (s_0, \{s_0, s_2\}), (s_0, \{s_0, s_1\})$. For ease of reading, Fig.~\ref{fig:ex-H2} does not show the outgoing edges from $(s_0, \{s_0, s_2\})$, $(s_0, \{s_0, s_1\})$, $(s_1, \{s_0, s_1\})$, $(s_1, \{s_1\})$ and $(s_2, \{s_2\})$. However, it can be verified that their outgoing edges visit $q_F$ or a state among the same set of states. Thus, by choosing the action $(a_0, \sense_0)$ within this subset of \ac{mdp}, P1 is guaranteed to visit $q_F$ with probability 1 \cite{puterman2014markov}.

    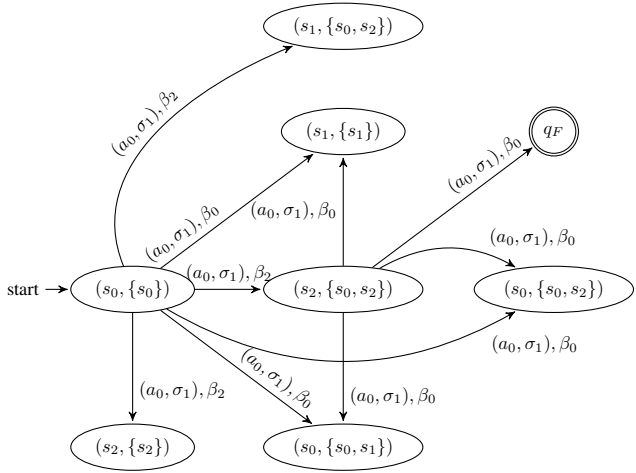
\begin{figure}
        \centering
        \begin{tikzpicture}[->,>=stealth',shorten >=1pt,auto,node distance=4cm,
        scale = 0.7, transform shape]
        \tikzstyle{every state}=[fill=white,draw,ellipse]

  \node[state, initial ] (0) {$(s_0, \{s_0\})$};
  \node[state] (1) [right of=0] {$(s_2, \{s_0, s_2\})$};
  \node[state ] (4) [right of=1] {$(s_0, \{s_0, s_2\})$};
  \node[state ] (3) [below of=1, node distance =3cm] {$(s_0, \{s_0, s_1\})$};
  \node[state ] (5) [above of=1, node distance =3cm] {$(s_1, \{s_1\})$};
  \node[state, accepting] (7) [right of=5] {$q_F$};
  \node[state] (6) [below of=0, node distance =3cm] {$(s_2, \{s_2\})$};
  \node[state] (2) [above of=5, node distance =2cm] {$(s_1, \{s_0, s_2\})$};

  \path 
    (0) edge  node[pos=0.7,sloped] {$(a_0, \sense_1), \att_0$} (3)
    (0) edge  node[sloped] {$(a_0, \sense_1), \att_2$} (1)
    (0) edge  node[sloped,pos=0.2] {$(a_0, \sense_1), \att_0$} (5)
    (0) edge[bend left, out=60]  node[sloped] {$(a_0, \sense_1), \att_2$} (2)
    (0) edge[bend right]  node[below,xshift=1cm, pos=0.9] {$(a_0, \sense_1), \att_0$} (4)
    (0) edge  node[pos=0.7] {$(a_0, \sense_1), \att_2$} (6)
    (1) edge  node[pos=0.75]{$(a_0, \sense_1), \att_0$} (3)
    (1) edge  node  {$(a_0, \sense_1), \att_0$} (5)
    (1) edge[bend left]  node[pos=0.75]{$(a_0, \sense_1), \att_0$} (4)
    (1) edge  node[pos=0.75,sloped] {$(a_0, \sense_1), \att_0$} (7)
  ;

\end{tikzpicture}
        \caption{A subset of perfect-observation \ac{mdp} constructed by P1 in her mind. The figure shows the relevant states and edges when P1 chooses the action $(a_0, \sense_1)$ at initial $q_0 = (s_0, \{s_0\})$ and either the sensor $0$ or $2$ fails.}
        \label{fig:ex-H2}
    \end{figure}
\end{example}

\subsection{The Game Perceived by P2}\label{subsec:p2game}

Unlike P1, P2 has perfect observability and is also aware of P1's misinterpretation of the sensor failures. Thus, in addition to keeping track of the true state of the game, he can also track P1's belief state. The resulting game model for P2 can be represented as follows. 

\begin{definition}
    Given P1's almost-sure winning, observation-based  strategy $\pi$ in $H$ such that P1 selects any action in $\pi(s,B)$ with a positive probability at a belief state $B$, P2's sensor attack strategy can be computed from qualitative planning in the following \ac{mdp} given as the tuple,
    \[
        G_2 = \langle \win_1, \attackAct, \delta_A, q_0, \win^1 \cap ((S\setminus F) \times 2^S) \rangle,
    \]
    where $\win_1$ is the set of almost-sure winning states of P1 in $H$. Given states $q=(s,B), q'=(s',B') \in \win_1$ and an attack action $\att \in \attackAct$, we have $\delta_A(q, \att, q') > 0$ if and only if there exists $(a, \sense) \in \pi(q)$ such that $s' \in \post_{G}(s, a)$ and $B' = \post_G(B, a) \cap \obs(s', \sense, \att)$. P2's objective is to ensure the game  staying in the set  $\win^1 \cap ((S\setminus F) \times 2^S)$ with probability one.
\end{definition}

A transition in $G_2$ is interpreted as follows. For a given state $(s,B) \in \win_1$, P1 thinks that this belief $B$ is almost-sure winning for her. Therefore, she chooses any action $(a,\sense) \in  \pi(s,B)$ with a non-zero probability. However, a state $(s',B')$ in $G_2$ will be reached probabilistically but partially controlled by P2: state $s'$ will be reached with a positive probability if $s' \in \post_G(s, a)$, the belief $B'$ is decided jointly by the sensing action and attack action, $B' = \post_G(B, a)\cap \obs(s',\sense,\att)$. It is important to note that this belief update is controlled partially by P2 given his choice of attack action. As P2's game $G_2$ is a perfect-observation game, the set of his almost-sure winning states to prevent P1 from reaching $F$ can be computed using the classical algorithm \cite[Alg.~46]{baier2008principles}.

Our interest in the P2's almost-sure winning region is to identify if there exist any states which P1 misinterprets to be almost-sure winning for her, but they are in fact almost-sure winning for P2.
This analysis also provides a way to detect adversarial attacks, for example, policy inference  or behavior cloning algorithms \cite{behaviorCloning2018} from the observed sensor failure data can be used to infer the ``rule'' behind the sensor failures and compare it with the attack policy. 

We illustrate the existence of such a state   using a simpler example shown in Fig.~\ref{fig:P2-winG}. The example consists of 4 states covered by three sensors indexed $0$ (blue), $1$ (red) and $2$ (green) covering the states $\gamma(0) = \{s_1\}$, $\gamma(1) = \{s_0, s_1\}$, and $\gamma(2) = \{s_2, s_3\}$, respectively. P1 has two sensing actions: $\sense_0$ and $\sense_1$ which query the sensors $\{1, 2\}$ and $\{1, 3\}$, respectively. P2 has three attack actions: $\att_j$ which jams the sensor $j$, for $j=0, 1, 2$. The corresponding P2's game is shown in Fig.~\ref{fig:P2-winH} (edges corresponding to $\sense_2, \att_2$ are omitted for clarity). Consider the states $(s_0, \{s_0\})$ and $(s_0, \{s_0, s_1\})$, at which P1's strategy is to select the action $(a_0, \sense_0)$. If P2 always selects $\att_0$ at these states, then the game is restricted within the states $(s_0, \{s_0\})$, $(s_0, \{s_0, s_1\})$ and $(s_1, \{s_0, s_1\})$, indefinitely. In other words, the states $(s_0, \{s_0\})$, $(s_0, \{s_0, s_1\})$ and $(s_1, \{s_0, s_1\})$, which P1 considers almost-sure winning for her, are in fact almost-surely winning for P2. P1's task failure is because she mistakes that sensor failure caused by $\beta_1$ is possible and sensor failure caused by $\beta_0$ cannot be persistent.

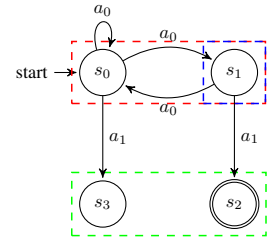
\begin{figure}
    \centering
    \begin{tikzpicture}[->,>=stealth',shorten >=1pt,auto,node distance=2.5cm,
        scale = 0.7,transform shape]

  \node[state, initial] (0) {$s_0$};
  \node[state] (1) [right of=0] {$s_1$};
  \node[state, accepting] (2) [below of=1] {$s_2$};
  \node[state] (3) [below of=0] {$s_3$};
    \node [container,fit=(0) (1),draw=red,dashed,line width=0.2mm  ] (container) {};
    \node [container,fit=(1),draw=blue,dashed,line width=0.2mm ] (container) {};
    \node [container,fit=(2) (3) ,draw=green, dashed,line width=0.2mm] (container) {};
  \path (0) edge[loop above]  node {$a_0$} (0)
        (0) edge[bend left]  node {$a_0$} (1)
        (1) edge[bend left]   node {$a_0$} (0)
        (0) edge   node {$a_1$} (3)
        (1) edge  node {$a_1$} (2)
  ;

\end{tikzpicture}
    \caption{A two-player stochastic game with partially controllable observation function. The dashed regions represent the sensors: $0$ (red), $1$ (blue), $2$ (green).}
    \label{fig:P2-winG}
\end{figure}

\begin{figure}
    \centering
    \begin{tikzpicture}[->,>=stealth',shorten >=1pt,auto,node distance=4cm,
        scale = 0.7,transform shape]
        \tikzstyle{every state}=[fill=white,draw,ellipse]

  \node[state, initial,fill=red!20] (0) {$s_0, \{s_0\}$};
  \node[state,fill=red!20] (1) [below right of=0] {$s_0, \{s_0, s_1\}$};
  \node[state,fill=red!20] (2) [above right of=1] {$s_1, \{s_0, s_1\}$};
  \node[state] (3) [below left of=1] {$s_1, \{s_1\}$};
  \node[state, accepting] (4) [below right of=1] {$q_F$};

  \path 
        (0) edge[loop above]  node {$(a_0, \sense_0), \att_1$} (0)
        (0) edge[bend right]  node {$(a_0, \sense_0), \att_1$} (3)
        (0) edge[bend right]  node[sloped] {$(a_0, \sense_0), \att_0$} (1)
        (0) edge[bend left]  node[sloped] {$(a_0, \sense_0), \att_0$} (2)
        (1) edge[bend right]  node[sloped,below] {$(a_0, \sense_0), \att_1$} (0)
        (1) edge[loop right]  node {$(a_0, \sense_0), \att_0$} (1)
        (1) edge  node[sloped] {$(a_0, \sense_0), \att_0$} (2)
        (2) edge  node[above] {$(a_0, \sense_0), \att_0$} (0)
        (1) edge  node {$(a_0, \sense_0), \att_1$} (3)
        (3) edge[below]  node {$(a_1, \cdot), \cdot$} (4)
  ;

\end{tikzpicture}
    \caption{P2's perfect-observation \ac{mdp} game corresponding to the game in Fig.~\ref{fig:P2-winG}. All states  in this figure are considered by P1 to be almost-sure winning. But the shaded, red states are P2's almost-sure winning states.}
    \label{fig:P2-winH}
\end{figure}
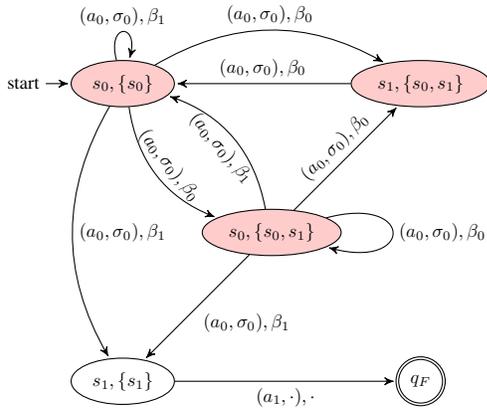

\section{Discussion and Conclusion}
For the class of stochastic games with partially controllable observation function in which P1 has partial and P2 has perfect observation, we presented a method to identify conditions under which the attacker has an almost-sure winning, sensor-attack strategy but when the system (a na\"ive player) considers to be her almost-sure winning state due to incorrectly treat sensor failures as probabilistic events.

With such an understanding of sensor-attack strategy, we will investigate the design of secured sensing and control strategy for a smart agent  (P1) who is aware of the adversarial sensor attacks. We will also be interested in  strategic  sensor design which ensures, despite a na\"ive P1, the attacker's almost-sure winning region is minimized, or does not include the initial game state. Quantitative planning under sensor attacks with formal methods  will also be investigated.

\bibliographystyle{IEEEtran}
\bibliography{refs}

\end{document}